\documentclass{lmcs}
\pdfoutput=1

\usepackage{lastpage}
\lmcsdoi{15}{1}{1}
\lmcsheading{}{\pageref{LastPage}}{}{}%
{Dec.~22,~2016}{Jan.~09,~2019}{}

\usepackage{hyperref}
\usepackage{MnSymbol}
\usepackage{cite,bbm,tikz}
\usepackage{ccs} 
\usetikzlibrary{arrows}

\keywords{Monad; monad composition; rewriting}

\ACMCCS{\printccsdesc}

\ccsdesc[500]{Theory of computation~Equational logic and rewriting}
\ccsdesc[500]{Theory of computation~Categorical semantics}

\newcommand\Set{\ensuremath{\mathsf{Set}}}
\newcommand\Cat{\ensuremath{\mathsf{Cat}}}
\newcommand\cod{\mathop{\mathsf{cod}}}
\newcommand\dom{\mathop{\mathsf{dom}}}
\newcommand\id{\mathsf{id}}

\newcommand\muP{\ensuremath{\mu^P}}
\newcommand\etaP{\ensuremath{\eta^P}}
\newcommand\muT{\ensuremath{\mu^T}}
\newcommand\etaT{\ensuremath{\eta^T}}
\newcommand\muPT{\ensuremath{\mu^{PT}}}
\newcommand\etaPT{\ensuremath{\eta^{PT}}}

\newcommand\CatB{\ensuremath{\mathsf B}}
\newcommand\CatC{\ensuremath{\mathsf C}}
\newcommand\CatD{\ensuremath{\mathsf D}}
\newcommand\CatE{\ensuremath{\mathsf E}}
\newcommand\derives{\stackrel *\to}

\newcommand\TCAT[2]{\ensuremath{(#1;#2)}}
\newcommand\RCAT[4]{(#1,#2,#3,#4)}
\newcommand\goesto[1]{\stackrel{#1}{\longrightarrow}}

\newcommand\eps\varepsilon
\renewcommand\star{^*}
\renewcommand\phi\varphi

\renewcommand\AA{\mathcal{A}}
\newcommand\CC{\mathcal{C}}
\newcommand\EE{\mathcal{E}}
\newcommand\FF{\mathcal{F}}
\newcommand\OO{\mathcal{O}}
\newcommand\RR{\mathcal{R}}

\newcommand\seq[3]{#1_{#2},\ldots,#1_{#3}}
\newcommand\set[2]{\{#1 \mid #2\}}
\newcommand\subs{\mathrel{\subseteq}}

\begin{document}

\title[Natural Transformations as Rewrite Rules]{Natural Transformations as Rewrite Rules\\and Monad Composition}

\author[D.~Kozen]{Dexter Kozen}
\address{Computer Science Department, Cornell University, Ithaca, NY 14853-7501, USA}
\email{kozen@cs.cornell.edu}

\begin{abstract}
\noindent Eklund et al.~(2002) present a graphical technique aimed at simplifying the verification of various category-theoretic constructions, notably the composition of monads. In this note we take a different approach involving string rewriting. We show that a given tuple $(T,\mu,\eta)$ is a monad if and only if $T$ is a terminal object in a certain category of strings and rewrite rules, and that this fact can be established by proving confluence of the rewrite system. We illustrate the technique on the monad composition problem. We also give a characterization of adjunctions in terms of rewrite categories.
\end{abstract}
\dedicatory{In honor of Ji\v{r}\'{i} Ad\'{a}mek on the Occasion of his Seventieth Birthday.}

\maketitle

\section{Introduction}
\label{sec:intro}

\nocite{Adamek09}%
As common constructions in the theory of data types and programming language semantics become better understood, there is a natural tendency toward generality. One desires to isolate common underlying principles, to unify related notions in a common framework, and to provide powerful abstract tools for reasoning and understanding.

A good example of one successful such enterprise is the use of monads in functional and logic programming \cite{Moggi91,Wadler95,BarrWells90,Eklund00}. Monads provide a clean way to combine modules or extend functionality of programming languages or data structures with new features such as continuations, state, and concurrency \cite{Moggi91,HarrisonKamin98,Wadler95,Claessen99}. They have been applied to parsing and type checking \cite{Wadler92}, semantics of nondeterministic and probabilistic computation \cite{Moggi91,RamseyPfeffer02,Panangaden00}, and unification in logic programming \cite{RydeheardBurstall86}.

Unfortunately, greater abstraction is often accompanied by reduced accessibility. Many abstract constructions, although well motivated by applications, may at times be difficult to navigate when presented in more abstract form. In particular, reasoning about the basic properties of monads---such as \emph{monad composition}, the construction that underlies many of the applications above---relies on the combinatorial manipulation of functors and natural transformations. Verification often requires a complicated process of arrow chasing in large diagrams. Specialized verification tasks such as the following example from \cite{BarrWells02} are not uncommon.
\begin{center}
  \small
  \begin{tikzpicture}[->, >=stealth', auto]
   \scriptsize
   \node (00) {$T_1^2T_2$};
   \node (01) [right of=00, node distance=35mm] {$T_1T_2T_1T_2T_1$};
   \node (02) [right of=01, node distance=25mm] {$T_1T_2T_1$};
   \node (10) [below of=00, node distance=1.3cm] {$T_1^2T_2T_1$};
   \node (11) [right of=10, node distance=35mm] {$T_2T_1T_2T_1T_2T_1$};
   \node (12) [right of=11, node distance=25mm] {$T_2T_1T_2T_1$};
   \node (20) [below of=10, node distance=1.3cm] {$T_1T_2T_1$};
   \node (21) [right of=20, node distance=35mm] {$T_2T_1T_2T_1$};
   \node (22) [right of=21, node distance=25mm] {$T_2T_1$};
   \path (00) edge node {$T_1\eta_2T_1T_2\eta_1$} (01)
         edge node[swap] {$T_1^2T_2\eta_1$} (10);
   \path (01) edge node {$T_1\mu$} (02)
         edge node[swap] {$\eta_2T_1T_2T_1T_2T_1$} (11);
   \path (02) edge node {$\eta_2T_1T_2T_1$} (12);
   \path (10) edge node[swap] {$\eta_2T_1\eta_2T_1T_2T_1$} (11)
         edge node[swap] {$\mu_1T_2T_1$} (20);
   \path (11) edge node[swap] {$T_2T_1\mu$} (12)
         edge node[swap] {$\mu T_2T_1$} (21);
   \path (12) edge node {$\mu$} (22);
   \path (20) edge node[swap] {$\eta_2T_1T_2T_1$} (21);
   \path (21) edge node[swap] {$\mu$} (22);
  \end{tikzpicture}
\end{center}
There are several general accounts of monad composition in the literature \cite{JonesDuponcheel93,BarrWells90,Eklund00}. In particular, Eklund et al.~\cite{Eklund00} make note of the difficulty of monad composition proofs and present a graphical technique aimed at simplifying the process. Their technique provides a pictorial representation of various constructions.

In this note we take a different approach involving string rewriting. We show that a given tuple $(T,\mu,\eta)$ is a monad if and only if $T$ is a terminal object in a certain category whose objects are strings and whose morphisms are rewrite rules, and the fact that $T$ is terminal can be established by proving confluence of the string rewriting system. We illustrate the technique on the monad composition problem studied by Eklund et al. We also give a characterization of adjoint functors in terms of rewrite categories.

\section{Preliminaries}
\label{sec:func}

Let $\CatC,\CatD$ be categories. Recall that a \emph{natural transformation} $\phi:F\to G$ between functors $F,G:\CatC\to\CatD$ is a collection of morphisms $\phi_A:FA\to GA$ of $\CatD$, one for each object $A$ of $\CatC$, such that for any morphism $h:A\to B$ of $\CatC$, the following diagram commutes:
\begin{equation}
\begin{array}c
  \begin{tikzpicture}[->, >=stealth', node distance=16mm, auto]
   \scriptsize
   \node (00) {$FA$};
   \node (01) [right of=00] {$FB$};
   \node (10) [below of=00, node distance=12mm] {$GA$};
   \node (11) [right of=10] {$GB$};
   \path (00) edge node {$Fh$} (01)
         edge node[swap] {$\phi_A$} (10);
   \path (10) edge node[swap] {$Gh$} (11);
   \path (01) edge node {$\phi_B$} (11);
  \end{tikzpicture}
\end{array}
\label{dia:nattrans}
\end{equation}
Categories, functors, and natural transformations form a 2-category \Cat\ in which the 0-cells (objects) are the categories, the 1-cells (morphisms) are the functors, and the 2-cells (morphisms of morphisms) are the natural transformations.

In a general 2-category, composition of 1-cells and composition of 2-cells are called \emph{horizontal} and \emph{vertical composition}, respectively, and are often denoted $\circ_0$ and $\circ_1$, respectively; but we will usually write $\circ$ for $\circ_1$ and omit the $\circ_0$ altogether. Horizontal composition also acts on 2-cells and satisfies the property:
If $F_1,F_2,F_3:\CatC\to\CatD$ and $G_1,G_2,G_3:\CatD\to\CatE$ are 1-cells and $\phi:F_1\to F_2$, $\phi':F_2\to F_3$, $\psi:G_1\to G_2$, and $\psi':G_2\to G_3$ are 2-cells, then
\begin{align}
(\psi'\circ\psi)(\phi'\circ\phi) &= (\psi'\phi')\circ(\psi\phi).\label{eq:2cell}
\end{align}

A special case of the action of horizontal composition on 2-cells is the following. If $S:\CatD\to\CatE$, $F,G:\CatC\to\CatD$, and $T:\CatB\to\CatC$ are 1-cells and $\phi:F\to G$ is a 2-cell, there is a 2-cell $S\phi T:SFT\to SGT$ obtained from the horizontal composition of $\phi$ with the identities $\id_S$ on the left and $\id_T$ on the right. It is helpful to think of this operation as a kind of scalar multiplication. It satisfies the following properties:
\begin{gather}
I\phi\ =\ \phi\ =\ \phi I\nonumber\\
(SS')\phi\ =\ S(S'\phi) \qquad \phi(TT')\ =\ (\phi T)T'\nonumber\\
S(\phi\circ\psi)\ =\ S\phi\circ S\psi \qquad (\phi\circ\psi)T\ =\ \phi T\circ\psi T.\label{eqn:module3}
\end{gather}
Thus we can write $SS'\phi$ and $\phi TT'$ without ambiguity.

\begin{exa}
In the 2-category \Cat, if $S$ and $T$ are functors and $\phi$ is a natural transformation with components $\phi_A$, then $S\phi T$ is the natural transformation with components $(S\phi T)_A=S(\phi_{TA})$.
\end{exa}

It follows from \eqref{eqn:module3} that the commutativity of diagrams is preserved under scalar multiplication; that is, if the left-hand diagram in \eqref{dia:extend} commutes, then so does the right:
\begin{equation}
\begin{array}c
  \begin{tikzpicture}[->, >=stealth', node distance=16mm, auto]
   \scriptsize
   \node (00) {$F$};
   \node (01) [right of=00] {$G$};
   \node (10) [below of=00, node distance=12mm] {$X$};
   \node (11) [right of=10] {$Y$};
   \path (00) edge node {$\phi$} (01)
         edge node[swap] {$\psi$} (10);
   \path (10) edge node[swap] {$\sigma$} (11);
   \path (01) edge node {$\tau$} (11);
   \node (00) [right of=01, node distance=25mm] {$SFT$};
   \node (01) [right of=00] {$SGT$};
   \node (10) [below of=00, node distance=12mm] {$SXT$};
   \node (11) [right of=10] {$SYT$};
   \path (00) edge node {$S\phi T$} (01)
         edge node[swap] {$S\psi T$} (10);
   \path (10) edge node[swap] {$S\sigma T$} (11);
   \path (01) edge node {$S\tau T$} (11);
  \end{tikzpicture}
\end{array}
\label{dia:extend}
\end{equation}

From these observations, one can begin to see the motivation for viewing natural transformations as rewrite rules: a rewrite rule $\phi:F\to G$ can be applied in the context of a string $S\phi T:SFT\to SGT$. For this to make sense, though, it had better be the case that the diamond property holds for non-overlapping redexes; that is, it must be possible to apply two rewrite rules with non-overlapping redexes in either order. We will see later (Lemma \ref{lem:disjoint}) that this is indeed the case, and in fact holds in all 2-categories as a consequence of \eqref{eq:2cell}.

\section{Rewrite Categories}

A \emph{rewrite category} is like a string rewrite system, except the semantics of the derivations are taken into account in the definitions of confluence and local confluence.

Formally, a \emph{rewrite category} is defined in terms of a finitely presented 2-category $\RCAT\OO\FF\RR\EE$, where $\OO$, $\FF$, and $\RR$ are finite sets of 0-cells, 1-cells, and 2-cells, respectively, and $\EE$ is a finite set of well-typed equations between 2-cell-valued expressions over $\FF$ and $\RR$ generating an equational theory on 2-cells. The 1-cells do not satisfy any equations except those imposed by the axioms of category theory, thus form the free category (free typed monoid) freely generated by $\FF$. Each 1-cell corresponds to a well-typed string $F_1\cdots F_n$, where \emph{well-typed} means $\dom F_i=\cod F_{i+1}$, $1\leq i\leq n-1$. We denote by $\FF\star$ the set of all well-typed strings over $\FF$. The elements of $\RR$ are called \emph{rewrite rules} and generate the set of all 2-cells under horizontal ($\circ_0$) and vertical ($\circ_1$) composition.

The rewrite category itself is a substructure of the category of 1- and 2-cells of this finitely presented 2-category. It is specified by a tuple $\TCAT\AA{\RR'}$, where $\RR'\subs\RR$ and $\AA$ is any subset of $\FF\star$ closed under the action of $\RR'$.

Rewrite categories provide a syntax that can be interpreted in the 2-category \Cat\ of categories, functors, and natural transformations. A well-typed string $F_1\cdots F_n$ is interpreted as a functor $\dom F_n\to\cod F_1$, where concatenation of strings is interpreted as composition of functors. The rewrite rules describe the action of natural transformations on these strings.

The equations $\EE$ generating the equational theory of the 2-category typically represent an equational axiomatization of the construct under study. However, from the perspective of string rewriting, they serve a different purpose: they represent local confluence conditions that can be used to establish global confluence in the rewrite category, leading to the existence of terminal objects.

\begin{exa}
\label{ex:monad}
Recall that a \emph{monad} on a category $\CatC$ is a triple $(T,\mu,\eta)$, where $T$ is an endofunctor on $\CatC$ and $\mu:T^2\to T$ and $\eta:I\to T$ are natural transformations such that $\mu\circ\mu T=\mu\circ T\mu$ and $\mu\circ\eta T=\mu\circ T\eta=\id$; that is, the diagrams
\begin{equation}
\begin{array}{c@{\hspace{1cm}}c}
  \begin{tikzpicture}[->, >=stealth', node distance=16mm, auto]
   \scriptsize
   \node (00) {$T^3$};
   \node (01) [right of=00] {$T^2$};
   \node (10) [below of=00, node distance=12mm] {$T^2$};
   \node (11) [right of=10] {$T$};
   \path (00) edge node {$T\mu$} (01)
         edge node[swap] {$\mu T$} (10);
   \path (10) edge node[swap] {$\mu$} (11);
   \path (01) edge node {$\mu$} (11);
   \node (00) [right of=01, node distance=25mm] {$T$};
   \node (01) [right of=00] {$T^2$};
   \node (10) [below of=00, node distance=12mm] {$T^2$};
   \node (11) [right of=10] {$T$};
   \path (00) edge node {$T\eta$} (01)
         edge node[swap] {$\eta T$} (10)
         edge node {$\id$} (11);
   \path (10) edge node[swap] {$\mu$} (11);
   \path (01) edge node {$\mu$} (11);
  \end{tikzpicture}
\end{array}
\label{dia:monad}
\end{equation}
commute. A typical example is the \emph{powerset monad} $(P,\muP,\etaP)$ on \Set, where $PA$ is the powerset of $A$, $\etaP_A(x) = \{x\}$, and $\muP_A(\CC) = \bigcup\CC$.

Monads are characterized by a rewrite category $\TCAT{T\star}{\eta,\mu}$, where $T\star$ is the free monoid on one generator $T$ with rewrite rules $\eta:I\to T$ and $\mu:T^2\to T$. The underlying 2-category has one 0-cell. The equational theory on 2-cells is given by the defining equations for monads \eqref{dia:monad}.

The properties \eqref{dia:monad} completely axiomatize the equational theory of monads, but more importantly from the from the point of view of rewriting, they specify local confluence properties in the case of overlapping redexes. For example, the string $T^3$ contains two overlapping redexes for the rule $\mu$. The left-hand diagram of \eqref{dia:monad} specifies that this configuration satisfies the diamond property.

We will show later (Theorem \ref{thm:monad}) that $T$ is a terminal object in the rewrite category, and that this fact is equivalent to the axiomatization \eqref{dia:monad}.
\end{exa}

\begin{exa}
\label{ex:monads}
The composition of monads is characterized by a rewrite category
\begin{align*}
\TCAT{\{P,T\}\star}{\eta^P,\mu^P,\eta^T,\mu^T,\theta}, 
\end{align*}
where $\{P,T\}\star$ is the free monoid on two generators $P,T$. As in Example \ref{ex:monad}, the underlying 2-category has one 0-cell. The equational theory is given by the axiomatization of two monads on endofunctors $P$ and $T$, respectively, and a distributive law $\theta:TP\to PT$ connecting them. We will give more details in \S\ref{sec:monads}. Again, the axioms can be viewed as local confluence properties for overlapping redexes in the rewrite category.
\end{exa}

\begin{exa}
\label{ex:adjunction}
Adjunctions are characterized by a rewrite category $\TCAT{F(GF)\star\cup G(FG)\star}{\eta,\eps}$. In this example, the underlying 2-category has two 0-cells $\CatC,\CatD$ representing two categories, 1-cells $F:\CatC\to\CatD$ and $G:\CatD\to\CatC$ representing left and right adjoint functors, respectively, and 2-cells $\eta:I\to GF$ and $\eps:FG\to I$ representing the unit and counit of the adjunction, respectively.
\end{exa}

\subsection{Local Confluence and Normal Forms}

Given a 1-cell $SXT$ and a rule $\phi:X\to Y$ of a rewrite category, the scalar multiple $S\phi T:SXT\to SYT$ can be viewed as rewriting the indicated occurrence of $X$ to $Y$ in the context of the string $SXT$. This is called a \emph{reduction}. The \emph{redex} of the reduction $S\phi T$ is the substring $X$ of $SXT$. A sequence of reductions from $Y$ to $Z$ is called a \emph{derivation}. We write $\pi:Y\derives Z$ if $\pi$ is such a derivation. A string is in \emph{normal form} if no rules apply.

The rules $\RR$ of the rewrite category each have a type $X\to Y$ for various $X,Y\in\FF\star$, where $\dom X=\dom Y$ and $\cod X=\cod Y$. Applying these rules to strings in $\FF\star$, one can generate diagrams in the rewrite category. We say that the system $\RR$
\begin{itemize}
\item
is \emph{confluent} if for any two derivations $X\derives U$ and $X\derives V$, there is a word $W$ and derivations $U\derives W$ and $V\derives W$ such that the resulting diagram commutes;
\item
is \emph{locally confluent} if for any two single reductions $X\to U$ and $X\to V$, there is a word $W$ and derivations $U\derives W$ and $V\derives W$ such that the resulting diagram commutes; and
\item
has the \emph{diamond property} if for any two reductions $X\to U$ and $X\to V$, there is a word $W$ and reductions $U\to W$ and $V\to W$ such that the resulting diagram commutes. This might be a pushout in the rewrite category, but not necessarily.
\end{itemize}
Local confluence does not imply confluence, but the diamond property does. See \cite{BaaderNipkow98} for a thorough treatment of these concepts.

Note, however, that our reductions have semantic content as well as syntactic. In our definitions of confluence, local confluence, and the diamond property, it is not enough that two derivations derive the same word; the resulting diagrams must also commute in the rewrite category. We say that two derivations $X\derives Y$ are \emph{equivalent} if the composition of the 2-cells along the two paths are equal.

\begin{exa}
One of the defining properties for monads, namely the left-hand diagram of \eqref{dia:monad}, says that $\mu$ as a reduction rule can be applied to the string $T^3$ in two ways to obtain $T^2$: one way as $\mu T$ to the leftmost two occurrences of $T$ (the left arrow of the diagram) and the other as $T\mu$ to the rightmost (the top arrow), in both cases giving $T^2$. By applying $\mu$ again to the two occurrences of $T^2$, we obtain a commutative diamond. 

By \eqref{dia:extend}, we can compose on the left and right with any strings in $T\star$:
\begin{equation}
\begin{array}c
  \begin{tikzpicture}[->, >=stealth', node distance=30mm, auto]
   \scriptsize
   \node (00) {$T^{m+n+3}$};
   \node (01) [right of=00] {$T^{m+n+2}$};
   \node (10) [below of=00, node distance=12mm] {$T^{m+n+2}$};
   \node (11) [right of=10] {$T^{m+n+1}$};
   \path (00) edge node {$T^{m+1}\mu T^n$} (01)
         edge node[swap] {$T^m\mu T^{n+1}$} (10);
   \path (10) edge node[swap] {$T^m\mu T^n$} (11);
   \path (01) edge node {$T^m\mu T^n$} (11);
  \end{tikzpicture}
\end{array}
\label{dia:muext}
\end{equation}
This says that any two reductions involving the rewrite rule $\mu$ with overlapping redexes, applied anywhere in a string of length at least three, can be completed to a commutative diamond. 
\end{exa}

For nonoverlapping redexes, the diamond property always holds. This is a consequence of property \eqref{eq:2cell} of 2-categories.

\begin{lem}
\label{lem:disjoint}
Any two applications of rewrite rules with disjoint redexes can be applied in either order, and the resulting diagram commutes. That is, any well-typed diamond of the form
\begin{equation}
\begin{array}c
  \begin{tikzpicture}[->, >=stealth', node distance=25mm, auto]
   \scriptsize
   \node (00) {$PQRST$};
   \node (01) [right of=00] {$PQRYT$};
   \node (10) [below of=00, node distance=12mm] {$PXRST$};
   \node (11) [right of=10] {$PXRYT$};
   \path (00) edge node {$PQR\tau T$} (01)
         edge node[swap] {$P\sigma RST$} (10);
   \path (10) edge node[swap] {$PXR\tau T$} (11);
   \path (01) edge node {$P\sigma RYT$} (11);
  \end{tikzpicture}
\end{array}
\label{dia:disjoint}
\end{equation}
commutes, where $P,Q,R,S,T,X,Y$ are 1-cells and $\sigma:Q\to X$ and $\tau:S\to Y$ are 2-cells.
\end{lem}
\begin{proof}
It follows from property \eqref{eq:2cell} of 2-categories that if $\phi:F\to G$ and $\psi:M\to N$ are 2-cells and the following diagram is well-typed, then it commutes.
\begin{equation}
\begin{array}c
  \begin{tikzpicture}[->, >=stealth', node distance=16mm, auto]
   \scriptsize
   \node (00) {$FM$};
   \node (01) [right of=00] {$FN$};
   \node (10) [below of=00, node distance=12mm] {$GM$};
   \node (11) [right of=10] {$GN$};
   \path (00) edge node {$F\psi$} (01)
         edge node[swap] {$\phi M$} (10);
   \path (10) edge node[swap] {$G\psi$} (11);
   \path (01) edge node {$\phi N$} (11);
   \path (00) edge node {$\phi\psi$} (11);
  \end{tikzpicture}
\end{array}
\label{dia:disjoint2}
\end{equation}
(In fact, \eqref{eq:2cell} is equivalent to \eqref{eqn:module3} and \eqref{dia:disjoint2}.)
The diagram \eqref{dia:disjoint} is a special case of \eqref{dia:disjoint2} with the following substitutions: $PQ$ for $F$, $RST$ for $M$, $PX$ for $G$, $RYT$ for $N$, $P\sigma$ for $\phi$, and $R\tau T$ for $\psi$.
\end{proof}

The property \eqref{dia:nattrans} of natural transformations is a stronger form of \eqref{dia:disjoint2}.

\section{Applications}
\label{sec:main}

In this section we demonstrate the use rewrite categories in the verification of monad composition as presented by Eklund et al.~\cite{Eklund00}. We also give a characterization of monads and adjunctions in terms of rewrite categories.

The following lemma and its proof introduce our approach at a basic level.

\begin{lem}
\label{lem:basic}
Let $T^+$ denote the set of nonnull strings of $T$'s and $\mu$ the rewrite rule $T^2\to T$. The following are equivalent:
\begin{enumerate}
\item
$\mu$ satisfies the left-hand diagram of \eqref{dia:monad};
\item
$T$ is a terminal object in the rewrite category $\TCAT{T^+}\mu$.
\end{enumerate}
\end{lem}
\begin{proof}
(1) $\Rightarrow$ (2):
Combining Lemma \ref{lem:disjoint} with the observation \eqref{dia:muext}, we have that the rewrite system consisting of the single rule $\mu$ on strings $T^n$ for $n\geq 1$ satisfies the diamond property and is therefore confluent. It follows that any diagram starting from $T^n$, $n\geq 1$, and ending with the normal form $T$ commutes. Moreover, there exists a reduction sequence from any such $T^n$ to $T$. Thus there is a unique morphism $T^n\to T$ for $n\geq 1$, so $T$ is a terminal object.

(2) $\Rightarrow$ (1):
Conversely, if $T$ is a terminal object, then the left-hand diagram of \eqref{dia:monad} must commute, since there is a unique morphism $T^3\to T$.
\end{proof}

\subsection{Monads}
\label{sec:monad}

Now we add the rewrite rule $\eta$ to the mix. This rule can be used to introduce a new occurrence of $T$ anywhere in the string. In the presence of $\eta$, $T$ is no longer a normal form, and in fact normal forms in the strict sense of string rewriting no longer exist. Nevertheless, $T$ is still a terminal object. Moreover, $T^0=I$ can now be included, since there is an arrow $I\to T$.

\begin{thm}
\label{thm:monad}
Consider a rewrite category on 1-cells $T\star$ and rules $\mu:T^2\to T$ and $\eta:I\to T$. The following are equivalent:
\begin{enumerate}
\item
$\mu$ and $\eta$ satisfy \eqref{dia:monad};
\item
$T$ is a terminal object in the rewrite category $\TCAT{T\star}{\mu,\eta}$.
\end{enumerate}
\end{thm}
\begin{proof}
(1) $\Rightarrow$ (2):
Suppose we have a rewrite system in which the rules can be classified as either \emph{bad rules} or \emph{good rules} (in our application, a rule is \emph{bad} if it increases the length of the string, e.g.~$\eta$). Call a derivation \emph{good} if it uses only good rules. Suppose further that
\begin{enumerate}
\item[1.]
every pair of good reductions $X\to U$ and $X\to V$ can be completed to a diamond using only good reductions $U\to W$ and $V\to W$; and
\item[2.]
every pair of reductions $X\to U$ and $X\to V$, good or bad, are confluent using only good derivations $U\derives W$ and $V\derives W$.
\end{enumerate}
These conditions hold for $\TCAT{T\star}{\mu,\eta}$ under the assumption (1). We have already argued condition 1 for $\mu$ in Lemma \ref{lem:basic}. The right-hand diagram in \eqref{dia:monad} implies condition 2 by immediately inverting any application of $\eta$ whenever it is applied to a nonnull string. For example, consider applications of $\eta$ and $\mu$ to a substring $T^2$, where $\eta$ is applied between the two occurrences of $T$. The two redexes thus overlap.
\begin{align*}
\begin{array}c
  \begin{tikzpicture}[->, >=stealth', node distance=16mm, auto]
   \scriptsize
   \node (00) {$T^2$};
   \node (01) [right of=00] {$T$};
   \node (10) [below of=00, node distance=12mm] {$T^3$};
   \node (11) [right of=10] {};
   \path (00) edge node {$\mu$} (01)
         edge node[swap] {$T\eta T$} (10);
  \end{tikzpicture}
\end{array}
\end{align*}
The top arrow is good, but the left arrow is bad. However, using the right-hand diagram of \eqref{dia:monad}, the diagram can be completed using only good arrows:
\begin{align*}
\begin{array}c
  \begin{tikzpicture}[->, >=stealth', node distance=16mm, auto]
   \scriptsize
   \node (00) {$T^2$};
   \node (01) [right of=00] {$T$};
   \node (10) [below of=00, node distance=12mm] {$T^3$};
   \node (11) [right of=10] {$T^2$};
   \node (12) [right of=11] {$T$};
   \path (00) edge node {$\mu$} (01)
         edge node[swap] {$T\eta T$} (10)
         edge node {$\id$} (11);
   \path (10) edge node[swap] {$T\mu$} (11);
   \path (11) edge node[swap] {$\mu$} (12);
   \path (01) edge node {$\id$} (12);
  \end{tikzpicture}
\end{array}
\end{align*}

Now we argue that any system satisfying 1 and 2 is confluent. Let $X\derives U$ and $X\derives V$ be any two derivations involving good or bad rules. Starting from the apex $X$, move down the two derivations, adding good diamonds to the diagram in the case 1 and good confluent derivations in case 2. All new transitions added to the diagram are good. When done, there are no more exposed bad rules, and the diagram can be completed by filling in with good diamonds.

Thus any two derivations $T^n\derives T$ are confluent via good derivations, which must be of the form $T\derives T$. But the only good derivation $T\derives T$ is the identity. It follows that any two derivations $T^n\derives T
$ are equivalent; in other words, $T$ is a terminal object.

(2) $\Rightarrow$ (1):
Conversely, if $T$ is a terminal object of the rewrite category, then all diagrams starting with any $T^n$ and ending with $T$ must commute, in particular those of \eqref{dia:monad}, the defining conditions for monads.
\end{proof}


\begin{cor}
\label{cor:monad}
Let $\sigma$ be a 2-functor from the free 2-category on one 0-cell, one 1-cell $T$, and two 2-cells $\mu,\eta$ to $\Cat$, the 2-category of categories, functors, and natural transformations. Let $\EE$ be the equational theory on 2-cells induced by $\sigma^{-1}$. The following are equivalent:
\begin{itemize}
\item The equations \eqref{dia:monad} are a logical consequence of $\EE$.
\item $T$ is terminal in the rewrite category $\TCAT{T\star}{\mu,\nu}$ modulo $\EE$.
\item The image $(\sigma(T),\sigma(\mu),\sigma(\eta))$ is a monad.
\end{itemize}
\end{cor}

\subsection{Monad Composition}
\label{sec:monads}

Let $(P,\muP,\etaP)$ and $(T,\muT,\etaT)$ be monads on a category $\CatC$ connected by a \emph{distributive law} (or \emph{swapper} in the terminology of \cite{Eklund00}), which is a natural transformation $\theta:TP\to PT$ satisfying the following properties:
\begin{eqnarray}
&
\begin{array}c
  \begin{tikzpicture}[->, >=stealth', node distance=16mm, auto]
   \scriptsize
   \node (00) {$TP^2$};
   \node (01) [right of=00] {$PTP$};
   \node (02) [right of=01] {$P^2T$};
   \node (10) [below of=00, node distance=12mm] {$TP$};
   \node (12) [below of=02, node distance=12mm] {$PT$};
   \path (00) edge node {$\theta P$} (01)
         edge node[swap] {$T\muP$} (10);
   \path (10) edge node {$\theta$} (12);
   \path (01) edge node {$P\theta$} (02);
   \path (02) edge node {$\muP T$} (12);
   \node (20) [below of=10, node distance=12mm] {$T^2P$};
   \node (21) [right of=20] {$TPT$};
   \node (22) [right of=21] {$PT^2$};
   \path (20) edge node[swap] {$T\theta$} (21)
         edge node {$\muT P$} (10);
   \path (21) edge node[swap] {$\theta T$} (22);
   \path (22) edge node[swap] {$P\muT$} (12);
   \node (00) [right of=02, node distance=25mm, yshift=-6mm] {$TP$};
   \node (10) [below of=00, node distance=12mm] {$T$};
   \node (11) [right of=10] {$PT$};
   \path (00) edge node {$\theta$} (11);
   \path (10) edge node {$T\etaP$} (00)
         edge node[swap] {$\etaP T$} (11);
   \node (01) [right of=00] {$P$};
   \path (01) edge node[swap] {$\etaT P$} (00)
         edge node {$P\etaT$} (11);
  \end{tikzpicture}
\end{array}
\label{dia:distrib}
\end{eqnarray}

Distributive laws are discussed in depth in \cite[\S9.2]{BarrWells02}. A typical application is the construction of the \emph{complex algebra} of an algebra, whose elements are sets of elements of the original algebra. Here $P$ would be the powerset monad and $T$ the term monad of some variety of algebras, and $\theta$ takes a term of sets $t(\seq A1n)$ and turns it into a set of terms $\set{t(\seq a1n)}{a_i\in A_i,\ 1\leq i\leq n}$. These constructions are discussed in \cite{Brink93,Jipsen03}. Another example would be the combination of the additive and multiplicative monoid structures in semirings. Here $P$ would be the finite powerset monad and $T$ the free monoid construction.

In light of the theme of this note, it should be clear that the conditions in \eqref{dia:distrib} are nothing more than a way to establish local confluence in the case of overlapping redexes between $\theta$ and the monad operations. For example, the top rectangle of the left-hand diagram of \eqref{dia:distrib} handles overlapping redexes involving $\theta$ and $\muP$.

The two monads $P$ and $T$ can be combined as follows. Define
\begin{align}
\muPT &= \muP T\circ P^2\muT\circ P\theta T:(PT)^2\to PT &
\etaPT &= \etaP T\circ\etaT:I\to PT.\label{eq:PTdef}
\end{align}
Then $(PT,\muPT,\etaPT)$ is again a monad \cite{JonesDuponcheel93,BarrWells02,Eklund00}. We will verify this using Theorem \ref{thm:monad}. Most of the work is contained in the following lemma.

\begin{lem}
\label{lem:monadcomp}
Consider a rewrite category on 1-cells $\{P,T\}\star$ in a 2-category with a single 0-cell and rules $\muP:P^2\to P$, $\etaP:I\to P$, $\muT:T^2\to T$, $\etaT:I\to T$, and $\theta:TP\to PT$, such that $(\muP,\etaP)$ and $(\muT,\etaT)$ both satisfy the monad axioms \eqref{dia:monad} {\upshape(}with appropriate substitutions{\upshape)}. The following statements are equivalent:
\begin{enumerate}
\item
$\muP,\etaP,\muT,\etaT$, and $\theta$ satisfy the distributive laws \eqref{dia:distrib}.
\item
$PT$ is a terminal object in the rewrite category \TCAT{\{P,T\}\star}{\muP,\muT,\etaP,\etaT,\theta}.
\end{enumerate}
\end{lem}
\begin{proof}
(1) $\Rightarrow$ (2):
First we show that the rewrite system is confluent. Recall that a rule is \emph{bad} if it increases length, \emph{good} otherwise. The good rules are $\muP$, $\muT$, and $\theta$, and the bad rules are $\etaP$ and $\etaT$.
 
Every pair of good reductions can be completed to a good confluent diagram. If the redexes do not overlap, this follows from Lemma \ref{lem:disjoint}. For redexes that overlap, all cases are covered by the left-hand diagram of \eqref{dia:monad} and the left-hand diagram of \eqref{dia:distrib}.

Given any derivation $\pi:X\derives Y$, possibly containing bad reductions, produce a new derivation $\pi'$ as follows:
\begin{enumerate}
\item[1.]
Extend the derivation to derive $PT$.
\item[2.]
Rearrange the resulting derivation $X\derives Y\derives PT$ to obtain an equivalent derivation $\pi':X\derives PT$ in which all the bad rules are applied after all the good rules.
\end{enumerate}

Step 1 can be accomplished by first introducing an occurrence of $P$ and/or $T$ using $\etaP$ and $\etaT$ if necessary, then moving all occurrences of $P$ to the left of all occurrences of $T$ using $\theta$, then collapsing the $P$'s using $\muP$ and the $T$'s using $\muT$.

Step 2 can be done without increasing the length of the derivation. For any bad reduction followed immediately by a good reduction, if the symbol introduced by the bad reduction is not part of the redex of the good reduction, then the two reductions can be switched by Lemma \ref{lem:disjoint}.

Otherwise, the symbol introduced by the bad reduction is part of the redex of the good reduction. There are only six ways this can happen:
\begin{align*}
& T\goesto{T\etaT}{}T^2\goesto{\muT}{}T
&& T\goesto{\etaT T}{}T^2\goesto{\muT}{}T
&& T\goesto{T\etaP}{}TP\goesto{\theta}{}PT
\end{align*}
and symmetrically for $P$. The first two are equivalent to the identity by the right-hand diagram of \eqref{dia:monad} and can be deleted. The last is equivalent to $T\goesto{\etaP T}{} PT$ by the right-hand diagram of \eqref{dia:distrib}. We can continue this process until there are no more bad reductions occurring before good reductions in the derivation.

If $X$ contains at least one occurrence each of $P$ and $T$, then this must also be true of any string derived from $X$, since all rules preserve this property. But then $\pi'$ can contain no bad reductions at all. If it did, then the last reduction would be bad. But it is impossible to derive $PT$ from such a reduction, since it would have to come from a string of length one, and no such string can be derived from $X$.

By a similar argument, if $X\in P^+$, then $\pi'$ contains exactly one bad reduction to introduce $T$, and it occurs last in the derivation. This last reduction must be of the form $P\goesto{P\etaT} PT$. Similarly, if $X\in T^+$, the last reduction of $\pi'$ is of the form $T\goesto{\etaP T} PT$, and this is the only bad reduction in the derivation.

If $X=I$, a similar argument shows that $\pi'$ must be either
\begin{align}
I\goesto{\etaT}T\goesto{\etaP T}{}PT
&& I\goesto{\etaP}P\goesto{P\etaT}{}PT.\label{eqn:XI}
\end{align}

Now suppose we are given derivations $\pi:X\derives U$ and $\rho:X\derives V$, possibly using both good and bad rules. To show confluence of $\pi$ and $\rho$, it suffices to show that $\pi':X\derives PT$ and $\rho':X\derives PT$ are confluent.

As argued above, if $X$ contains at least one occurrence each of $P$ and $T$, then $\pi'$ and $\rho'$ are good. Thus we can complete them to a good commutative diagram by filling in with good commutative diamonds and \eqref{dia:distrib}. This process must terminate, since there is a fixed upper bound, quadratic in the length of $X$, on the length of any good derivation from $X$, since each good reduction strictly decreases the string in length or lexicographic order relative to $P < T$ \cite[Lemma 2.7.2]{BaaderNipkow98}.

If $X\in P^+$, then as argued above, $\pi'$ and $\rho'$ are both of the form $X\derives P\goesto{P\etaT} PT$, where the prefixes $X\derives P$ contain no occurrence of $T$. By Theorem \ref{thm:monad}, $P$ is terminal in $\TCAT{P\star}{\muP,\etaP}$, therefore the two prefixes $X\derives P$ are equivalent, and consequently so are $\pi'$ and $\rho'$. The argument is similar for $X\in T^+$.

When $X=I$, we need only observe that the two derivations \eqref{eqn:XI} form a commutative diamond by Lemma \ref{lem:disjoint}.

Finally we show that $PT$ is a terminal object. We have already argued that there is at least one derivation of $PT$ from every $X\in\{P,T\}\star$, so there is at least one arrow $X\to PT$ in the rewrite category. To show that there is at most one, let $\pi,\rho:X\derives PT$ be any two derivations. By confluence, we can complete to a pair of equivalent derivations $\pi',\rho':X\derives PT\derives PT$. Rearranging the final portions $PT\derives PT$ of these two derivations by Step 2 above, we obtain good derivations. But the only good derivation $PT\derives PT$ is the identity, therefore $\pi$ and $\rho$ were already equivalent.

(2) $\Rightarrow$ (1):
If $PT$ is a terminal object, then \eqref{dia:distrib} holds, as all maximal paths in all diagrams of \eqref{dia:distrib} lead to $PT$.
\end{proof}

\begin{thm}[\cite{JonesDuponcheel93,BarrWells02,Eklund00}]
\label{thm:monadcomp}
Let $(T,\muT,\etaT)$ and $(P,\muP,\etaP)$ be monads connected by a distributive law $\theta:TP\to PT$ satisfying \eqref{dia:distrib}.
Let $\muPT$ and $\etaPT$ be defined by \eqref{eq:PTdef}. Then $(PT,\muPT,\etaPT)$ is a monad.
\end{thm}
\begin{proof}
Consider the rewrite category\footnote{At the risk of ambiguity, we are overloading the symbols $P$, $\muP$, etc.~to refer to the components of the rewrite category as well as their images in \Cat.} \TCAT{\{P,T\}\star}{\muP,\muT,\etaP,\etaT,\theta}. By Theorem \ref{thm:monad}, it suffices to show that $PT$ is a terminal object in the rewrite category \TCAT{(PT)\star}{\muPT,\etaPT}. By Lemma \ref{lem:monadcomp}, we know that it is terminal in the category \TCAT{\{P,T\}\star}{\muP,\muT,\etaP,\etaT,\theta}, thus for any $n\geq 0$, there is exactly one arrow $(PT)^n\to PT$ in that category. It follows that there is at most one arrow $(PT)^n\to PT$ in the subcategory \TCAT{(PT)\star}{\muPT,\etaPT}. But there is also at least one, since we can derive $PT$ from $(PT)^0=I$ by the rule $\etaPT$ and from $(PT)^n$ for $n\geq 1$ by $n-1$ applications of the rule $\muPT$.
\end{proof}

\subsection{Adjunctions}
\label{sec:adjunction}

As a final application, we give a characterization of adjoint functors in terms of rewrite categories.

Recall that a functor $F:\CatC\to\CatD$ is a left adjoint of another functor $G:\CatD\to\CatC$ if there are natural transformations $\eta:I\to GF$ and $\eps:FG\to I$, the \emph{unit} and \emph{counit} of the adjunction, respectively, satisfying
\begin{align}
\begin{array}c
  \begin{tikzpicture}[->, >=stealth', node distance=16mm, auto]
   \scriptsize
   \node (00) {$G$};
   \node (10) [below of=00, node distance=12mm] {$GFG$};
   \node (11) [right of=10] {$G$};
   \path (00) edge node {$\id$} (11)
              edge node[swap] {$\eta G$} (10);
   \path (10) edge node[swap] {$G\eps$} (11);
   \node (00) [right of=00, node distance=25mm] {$F$};
   \node (01) [right of=00] {$FGF$};
   \node (11) [below of=01, node distance=12mm] {$F$};
   \path (00) edge node[swap] {$\id$} (11)
              edge node {$F\eta$} (01);
   \path (01) edge node {$\eps F$} (11);
  \end{tikzpicture}
\end{array}
\label{dia:adj}
\end{align}

\begin{thm}
\label{thm:adj}
Consider a 2-category on two 0-cells $\CatC$, $\CatD$ generated by 1-cells $F:\CatC\to\CatD$ and $G:\CatD\to\CatC$ and 2-cells $\eta:I\to GF$ and $\eps:FG\to I$. The following statements are equivalent:
\begin{enumerate}
\item
$\eta$ and $\eps$ satisfy \eqref{dia:adj}.
\item
$F$ is a terminal object in the rewrite category \TCAT{F(GF)\star}{\eta,\eps} and
$G$ is a terminal object in the rewrite category \TCAT{G(FG)\star}{\eta,\eps}.
\end{enumerate}
\end{thm}
\begin{proof}
(1) $\Rightarrow$ (2):
By arguments similar to the proof of Theorem \ref{thm:monad} and Lemma \ref{lem:monadcomp}, in any derivation $\pi:X\derives Y$, all applications of $\eta$ (the bad rule) can be moved after all applications of $\eps$ (the good rule), by either Lemma \ref{lem:disjoint} in the case of disjoint redexes or \eqref{dia:adj} in the case of overlapping redexes. The resulting equivalent derivation $\pi':X\derives Y$ is no longer than $\pi$.

Now if $\pi'$ is of the form $F(GF)^n\to F$, there can be no application of $\eta$, because the last $\eta$ would produce a string of length at least 2. But then all redexes are redexes of $\eps$, therefore are disjoint and can be done in any order by Lemma \ref{lem:disjoint}. Thus all derivations of the form $F(GF)^n\to F$ are equivalent. Moreover, there exists a derivation $F(GF)^n\to F$ consisting of $n-1$ applications of $\eps$. Thus $F$ is terminal in \TCAT{F(GF)\star}{\eta,\eps}. A symmetric argument shows that $G$ is terminal in \TCAT{G(FG)\star}{\eta,\eps}.

(2) $\Rightarrow$ (1):
In \eqref{dia:adj}, there is a unique morphism to $G$ in the left-hand diagram and a unique morphism to $F$ in the right-hand diagram, thus the diagrams commute.
\end{proof}

\begin{cor}
\label{cor:adj}
Let $\sigma$ be a 2-functor from the free 2-category on two 0-cells $\CatC,\CatD$, two 1-cells $F:\CatC\to\CatD$, $G:\CatD\to\CatC$, and two 2-cells $\eta:I\to GF$, $\eps:FG\to I$ to $\Cat$. Let $\EE$ be the equational theory on 2-cells induced by $\sigma^{-1}$. The following are equivalent:
\begin{itemize}
\item
The equations \eqref{dia:adj} are a logical consequence of $\EE$.
\item
$F$ is terminal in the rewrite category \TCAT{F(GF)\star}{\eta,\eps} modulo $\EE$ and $G$ is terminal in the rewrite category \TCAT{G(FG)\star}{\eta,\eps} modulo $\EE$.
\item
$\sigma(F)$ and $\sigma(G)$ are adjoint functors with $\sigma(F)$ the left adjoint and $\sigma(G)$ the right adjoint with unit $\sigma(\eta)$ and counit $\sigma(\eps)$.
\end{itemize}
\end{cor}

\section{Conclusion}

It is clear that these techniques are related to the notion of \emph{free adjunction} and \emph{free monad}
as presented by Schanuel and Street \cite{SchanuelStreet86}. A free adjunction is a certain 2-category that
is initial among all adjunctions and characterizes their equational theory.
For the future, we would like to explore these connections further and perhaps develop a notion of rewriting
for 2-categories and higher.

\section*{Acknowledgments}

Thanks to Riccardo Pucella for valuable comments.


\end{document}